\documentclass[10pt,a4paper]{article}

\usepackage[T1]{fontenc}
\usepackage[utf8]{inputenc}
\usepackage{lmodern}
\usepackage{microtype}
\usepackage{amsmath,amssymb,amsthm,mathtools,bm}
\usepackage{graphicx}
\usepackage{float}
\usepackage{booktabs,tabularx,array,multirow}
\usepackage{enumitem}
\usepackage[margin=2.25cm]{geometry}
\usepackage{caption}
\usepackage{fancyhdr}
\usepackage[hidelinks]{hyperref}
\hypersetup{
  pdftitle={Bellman--Shoreline Search in Arbitrary Dimension},
  pdfauthor={Florentin Koch},
  pdfsubject={Online search for unknown affine hyperplanes; exponential oscillators and effective computability},
  pdfkeywords={online search, hyperplane search, Shoreline search, self-similarity, support functions, active memory, computability, precession}
}
\usepackage[nameinlink,noabbrev]{cleveref}

\newcommand{\R}{\mathbb{R}}
\newcommand{\Sph}{\mathbb{S}}
\newcommand{\conv}{\operatorname{conv}}
\newcommand{\inrad}{\operatorname{inrad}_{0}}

\newcommand{\Var}{\operatorname{Var}}
\newcommand{\relint}{\operatorname{relint}}
\newcommand{\CR}{\operatorname{CR}}
\newcommand{\sgn}{\operatorname{sgn}}
\newcommand{\dd}{\,\mathrm{d}}
\newcommand{\e}{\mathrm{e}}
\newcommand{\Tphys}{T_{\mathrm{phys}}}
\newcommand{\tphys}{t_{\mathrm{phys}}}
\newcommand{\cwork}{C^{\mathrm{work}}}
\newcommand{\ccand}{C^{\mathrm{cand}}}
\newcommand{\cstar}{C^{*}}
\newcommand{\BD}{B_D}

\theoremstyle{plain}
\newtheorem{theorem}{Theorem}[section]
\newtheorem{proposition}[theorem]{Proposition}

\theoremstyle{definition}
\newtheorem{definition}[theorem]{Definition}
\theoremstyle{remark}
\newtheorem{remark}[theorem]{Remark}

\setlist{itemsep=.2em,topsep=.4em}
\title{\textbf{Bellman--Shoreline Search in Arbitrary Dimension}\\[0.35em]
\large Exponential Vector Oscillators, Active Memory, Precession, and Effective Computability}
\author{Florentin Koch\\
\small École Polytechnique\\
\small \texttt{florentin.koch@polytechnique.edu}}
\date{29 August 2026}

\begin{document}
\maketitle

\begin{abstract}
We study online search for an unknown affine hyperplane in $\R^D$, for arbitrary finite dimension. We build on the self-similar cell theorem, the effectively bounded return theorem, and the exact support-function formulation established by Koch~\cite{koch-cell}: after quotienting by scale, any finite-value history can be replaced, with arbitrarily small loss in minimax value, by the bilateral repetition of one cell, while the useful memory is the support function of the historical convex hull, updated by a maximum. The present article does not reprove that reduction. It asks what becomes of its mechanism when the direction space grows from $\Sph^0$ to $\Sph^{D-1}$.

The discussion starts in $D=1$. Alternation, productivity, and an adversary placed just beyond the latest record lead directly to an equal-ripple principle; the optimal stationary cell doubles at every turn and yields the exact constant $9$. In $D=2$, the same mechanism becomes continuous: a relative equilibrium is a logarithmic spiral, whose bottleneck chord imposes a tangency condition, after which the pitch is selected by minimizing the stationary ratio. These cases motivate the central object
\[
  \Gamma(\sigma)=\e^{\kappa\sigma}\omega(\sigma),
  \qquad \omega(\sigma)\in\Sph^{D-1}.
\]
We collect its log-directional geometry, exponentially discounted memory, gauges, and recursive hyperspherical parametrization. The contribution of the present article has four layers. First, without any shape ansatz, the bottleneck of an exponential orbit has a certificate supported by at most $D$ genuine historical suppliers; at globally worst phases, the current point belongs to the active face. Second, within each regular active chamber, the bottleneck variation, supplier tangencies, pitch identity, age bound, and---in $D=3$---a regular implicit delay system are exact. Third, the odd-dimensional obstruction, the antipodal subclass, and the harmonic tower provide constraints and explicit laboratories without being promoted to global minimizers. Finally, the N-COMP theorem shows that, for every fixed finite $D$, $\cstar_D$ is computable and an algebraic polygonal $\varepsilon$-optimal cell can in principle be synthesized. The numerical screening through $D=10$ is kept separate from these theorems.
\end{abstract}


\section{Introduction}

The Bellman--Shoreline problem asks a mobile agent starting at the origin to meet an affine hyperplane whose normal and distance are both unknown. A strategy is a locally rectifiable curve
\[
  \gamma:[0,\infty)\to\R^D,
  \qquad \gamma(0)=0,
\]
parametrized by arclength. The adversary chooses $u\in\Sph^{D-1}$ and $d>0$, hence
\[
  H(u,d)=\{x\in\R^D:\langle u,x\rangle=d\},
  \qquad
  \cstar_D=\inf_{\gamma}\sup_{u,d>0}\frac{T_\gamma(u,d)}{d}.
\]

The companion paper~\cite{koch-cell} treats the structural problem that precedes every shape conjecture: it removes the infinity of scales without assuming a spiral, an angular ordering, or internal stationarity. More precisely, its Theorems~4.1 and~5.1 provide qualitative and effective self-similar cell reductions, while its Proposition~7.1 gives the exact support--inradius formula. We use those results as a foundation.

\paragraph{Scope of the present article.}
This text does not reprove the cell reduction, nor does it claim a global optimizer in every dimension. It studies how the minimax mechanism changes when the normal space passes from two points to the sphere $\Sph^{D-1}$. Its logical chain is
\[
\begin{aligned}
  D=1 &\longrightarrow D=2
  \longrightarrow \Gamma(\sigma)=\e^{\kappa\sigma}\omega(\sigma)\\
  &\longrightarrow \text{active memory and bottleneck faces}\\
  &\longrightarrow \text{precession and candidate families}.
\end{aligned}
\]
The central new contribution is therefore not one particular curve, but a finite interface between a functional Bellman state and its worst bottleneck: a certificate of rank at most $D$, current-point contact, an envelope law, tangency, pitch, age, and supplier-delay equations. Precessing and harmonic families then serve as explicit laboratories constrained by these laws.

Three distinct compressions must not be conflated. The self-similar cell theorem compresses the minimax value to a repeated cell; sliding memory compresses the time horizon; Carathéodory's theorem on an exposed face compresses only an instantaneous bottleneck certificate. The third compression never makes the Bellman state finite-dimensional.

We separate four levels: exact identities, regular-chamber equations, explicit candidate families, and minimax conjectures. The spiral benchmark $13.811135\ldots$ is not identified with the global value, and no precessing family is declared optimal.

\section{Prerequisites from the companion paper}

This section records only the results from Koch~\cite{koch-cell} used below. It fixes notation and does not add a new proof step before the one-dimensional analysis.

\subsection{Support, common level, and the true online ratio}

For a trajectory $\gamma$, set
\[
  K_t=\conv\gamma([0,t]),
  \qquad
  h_t(u)=\max_{0\le s\le t}\langle u,\gamma(s)\rangle,
  \qquad
  m(t)=\min_{\|u\|=1}h_t(u)=\inrad(K_t).
\]
The competitive ratio is exactly
\[
  \CR(\gamma)=\sup_{t:m(t)>0}\frac{t}{m(t)}.
\]
This is Proposition~7.1 of~\cite{koch-cell}. Each normal therefore sees a one-dimensional record process, but all those records are generated by a single physical curve.

\subsection{Self-similar cells and memory by maximum}

The self-similar reduction of~\cite[Theorems~4.1 and~5.1]{koch-cell} identifies the infimal value with the infimum over bilateral repetitions of homothetic cells, up to arbitrarily small loss. The equivariant variant~\cite[Corollary~6.1]{koch-cell} allows a compact isometry:
\[
  \Gamma(s+T)=Qg\Gamma(s),
  \qquad Q>1,
\]
without imposing monotone rotation, a unique cycle, or stationarity inside the cell.

After quotienting by scale, one logarithmic step contracts the old convex hull by $a\in(0,1)$ and appends a fresh segment $c$:
\[
  H^+=\conv(aH\cup c),
  \qquad
  h^+(u)=\max\{a h(u),h_c(u)\}.
\]
This maximum law exactly preserves corners, returns, supplier changes, and pre-service.

\subsection{Sliding memory}

By~\cite[Theorem~12.2]{koch-cell}, if $\CR(\gamma)\le C$, then
\[
  K_t=\conv\gamma([t/C,t]).
\]
Thus useful history has the exact horizon $\log C$ in logarithmic time. This bounds memory age, but not the number of switches or the total period of a cell.

These are the three imported ingredients needed below. The next two propositions---equivalence of scale conventions and the equivariant cocycle---belong to the present paper. They connect the cell reduction to the higher-dimensional literature and prevent one from confusing cellular self-similarity with angular stationarity.

\subsection{Two technical bridges}

Hyperplane-search papers do not always use the same scale convention: some impose $d\ge1$, while others retain only an asymptotic coefficient. Homogeneity and bilateral cell repetition show that these choices do not alter the infimal value.

\begin{proposition}[Equivalence of scale conventions]
Set
\[
  C_D^{\mathrm{all}}
  =\inf_\gamma\sup_{u,d>0}\frac{T_\gamma(u,d)}{d},
  \qquad
  C_D^{\ge1}
  =\inf_\gamma\sup_{u,d\ge1}\frac{T_\gamma(u,d)}{d},
\]
and
\[
  C_D^{\mathrm{asym}}
  =\inf_\gamma\limsup_{d\to\infty}
  \sup_{u\in\Sph^{D-1}}\frac{T_\gamma(u,d)}{d}.
\]
Then
\[
  C_D^{\mathrm{all}}=C_D^{\ge1}=C_D^{\mathrm{asym}}.
\]
The same value is obtained if one allows an additive constant in $T_\gamma(u,d)\le Cd+B$ and minimizes only the coefficient $C$.
\end{proposition}

\begin{proof}
The inequalities $C_D^{\mathrm{asym}}\le C_D^{\ge1}\le C_D^{\mathrm{all}}$ are immediate. Conversely, fix $\eta>0$ and a history whose asymptotic coefficient is at most $C_D^{\mathrm{asym}}+\eta$. There exists $d_0$ such that every level $d\ge d_0$ is completed with quotient at most $C_D^{\mathrm{asym}}+2\eta$. The construction in~\cite[Theorem~4.1]{koch-cell} may be started on a geometric sequence lying entirely above $d_0$: neither the quasi-return, the connector, nor payment for the past uses lower levels. It produces a homothetic cell whose bilateral repetition has ratio at most $C_D^{\mathrm{asym}}+3\eta$. Negative copies recreate all small scales. Hence $C_D^{\mathrm{all}}\le C_D^{\mathrm{asym}}+3\eta$, and $\eta\downarrow0$ concludes. If $T\le Cd+B$, then $T/d\le C+B/d$; apply the same argument beyond $d_0\gg B/\eta$.
\end{proof}

An equivariant cell may return to the same shape after rotation and dilation without having monotone angle or being a spiral. The following cocycle records exactly what is periodic and isolates zero net winding.

\begin{proposition}[Equivariant cocycle, winding, and absence of secular drift]
Suppose in dimension two that an equivariant cell satisfies
\[
  \Gamma(s+T)=Q R_\Delta\Gamma(s),
  \qquad Q>1.
\]
On phases where $\Gamma(s)\neq0$, choose continuous lifts
\[
  \rho(s)=\log\|\Gamma(s)\|,
  \qquad
  \theta(s)=\arg\Gamma(s).
\]
There exists $\nu\in\mathbb Z$ such that
\[
  \rho(s+T)=\rho(s)+H,
  \qquad
  \theta(s+T)=\theta(s)+\Omega,
  \qquad
  H=\log Q,
  \quad
  \Omega=\Delta+2\pi\nu.
\]
If $\Omega\neq0$ and $\kappa=H/\Omega$, then
\[
  w(s):=\rho(s)-\kappa\theta(s)
  \qquad\text{satisfies}\qquad
  w(s+T)=w(s).
\]
Cell periodicity does not imply that $\theta$ is monotone or that $\rho$ is a single-valued function of $\theta$. The case $\Omega=0$ is not excluded by the cell theorem.
\end{proposition}

\begin{proof}
The similarity multiplies radius by $Q$ and adds $\Delta$ to the angle modulo $2\pi$, giving the two relations with a winding integer $\nu$. If $\Omega\neq0$, then
\[
  w(s+T)-w(s)=H-\frac{H}{\Omega}\Omega=0.
\]
A periodic function of cell phase may nevertheless contain angular returns, several radii at the same angle, and an online ordering distinct from the static ordering of the support front.
\end{proof}

From this point onward, all imported results are cited by their precise location in~\cite{koch-cell}. The first new object is the case in which the normal space contains only two points.

\section{Dimension one: the mechanism in full view}

Dimension one displays, without superfluous geometry, why alternation, equal ripple, and exponential growth appear.

\subsection{Alternation and productivity are forced}

Let $x_n$ be the successive turns that create new records. A finite-ratio strategy must serve both half-lines indefinitely. After deleting turns that create no new record, the signs alternate. Writing $a_n=|x_n|>0$, one may represent a step as
\[
  x_{n+1}=-\sgn(x_n)f_n(a_n),
\]
or absorb the sign into the notation and write $x_{n+1}=-f_n(x_n)$.

Alternation alone is not enough. When the searcher returns to the same side two turns later, it must exceed the earlier record:
\[
  a_{n+2}>a_n.
\]
Otherwise the full detour has created no new service on that half-line and may be deleted. This is the minimal productivity condition.

Let $S_n=\sum_{j=0}^n a_j$. The time at the turn establishing record $a_n$ is
\[
  t_n=2\sum_{j=0}^{n-1}a_j+a_n.
\]
Immediately after record $a_n$ is established, the adversary may place the target on the same side at distance $d=a_n+\varepsilon$. It is reached only after the complete excursion to the opposite side and the return. Letting $\varepsilon\downarrow0$ gives the exact scenario
\[
  R_n=1+2\frac{S_{n+1}}{a_n}.
\]
Thus the worst case is already visible just beyond the latest record that seemed secured.

\subsection{Writing both sides: equal ripple}

Consider a bilateral stationary cell without yet assuming the same factor at each half-turn:
\[
  a_{2k}=Aq^k,
  \qquad
  a_{2k+1}=Bq^k,
  \qquad q>1.
\]
The two asymptotic adversarial families are
\[
  R_+=1+2\frac{q(A+B)}{A(q-1)},
  \qquad
  R_-=1+2\frac{q(Aq+B)}{B(q-1)}.
\]
If $R_+>R_-$, too much has been invested on one side and too little on the other. The minimax balance therefore equalizes the two active ripples. The condition $R_+=R_-$ gives
\[
  B^2=A^2q.
\]
With $\rho=B/A=\sqrt q$, it follows that
\[
  \frac{a_{n+1}}{a_n}=\rho
  \qquad\text{for every }n.
\]
Equal ripple forces the same multiplicative advance on both sides.

\subsection{Why the optimal factor is two}

Once both sides are equalized, the problem becomes scalar. For $a_{n+1}=\rho a_n$,
\[
  R(\rho)=1+2\frac{\rho^2}{\rho-1},
  \qquad \rho>1.
\]
Hence
\[
  R'(\rho)=2\frac{\rho(\rho-2)}{(\rho-1)^2},
\]
and the unique minimum is attained at
\[
  \rho^*=2,
  \qquad
  R(2)=9.
\]
The factor two is therefore not a convention of the familiar zigzag; it is the repeated composition factor that equalizes the two sides and minimizes the worst target placed just beyond a record.

\subsection{Global closure: no nonstationary sequence does better}

The preceding argument explains the stationary mechanism. To certify global optimality among all productive zigzags, set
\[
  z_n=\frac{S_n}{a_n},
  \qquad
  \rho_{n+1}=\frac{a_{n+1}}{a_n}.
\]
Then
\[
  z_{n+1}=1+\frac{z_n}{\rho_{n+1}},
  \qquad
  R_n=1+2(z_n+\rho_{n+1}).
\]
If $\alpha:=\sup_n(z_n+\rho_{n+1})<4$, then $\rho_{n+1}\le\alpha-z_n$ and
\[
  z_{n+1}\ge1+\frac{z_n}{\alpha-z_n}=:F_\alpha(z_n).
\]
For $\alpha<4$,
\[
  F_\alpha(z)-z=\frac{z^2-\alpha z+\alpha}{\alpha-z}>0
\]
on $1\le z<\alpha$, with a strictly positive minimum. The sequence $z_n$ would therefore drift above $\alpha$, a contradiction. Consequently
\[
  \sup_n(z_n+\rho_{n+1})\ge4,
  \qquad
  \cstar_1=9.
\]

\subsection{The zigzag as an exponential oscillator}

At productive phases $\theta_n=n\pi$,
\[
  x_n=a_0\e^{\kappa_1\theta_n}\cos\theta_n,
  \qquad
  \kappa_1=\frac{\log2}{\pi}.
\]
The optimal zigzag is therefore a discrete exponential oscillator on $\Sph^0=\{-1,+1\}$. This interpretation, rather than the special geometry of a segment, is what will be transported to higher dimensions.

\begin{figure}[H]
  \centering
  \includegraphics[width=0.98\linewidth,trim=0 101bp 0 0,clip]{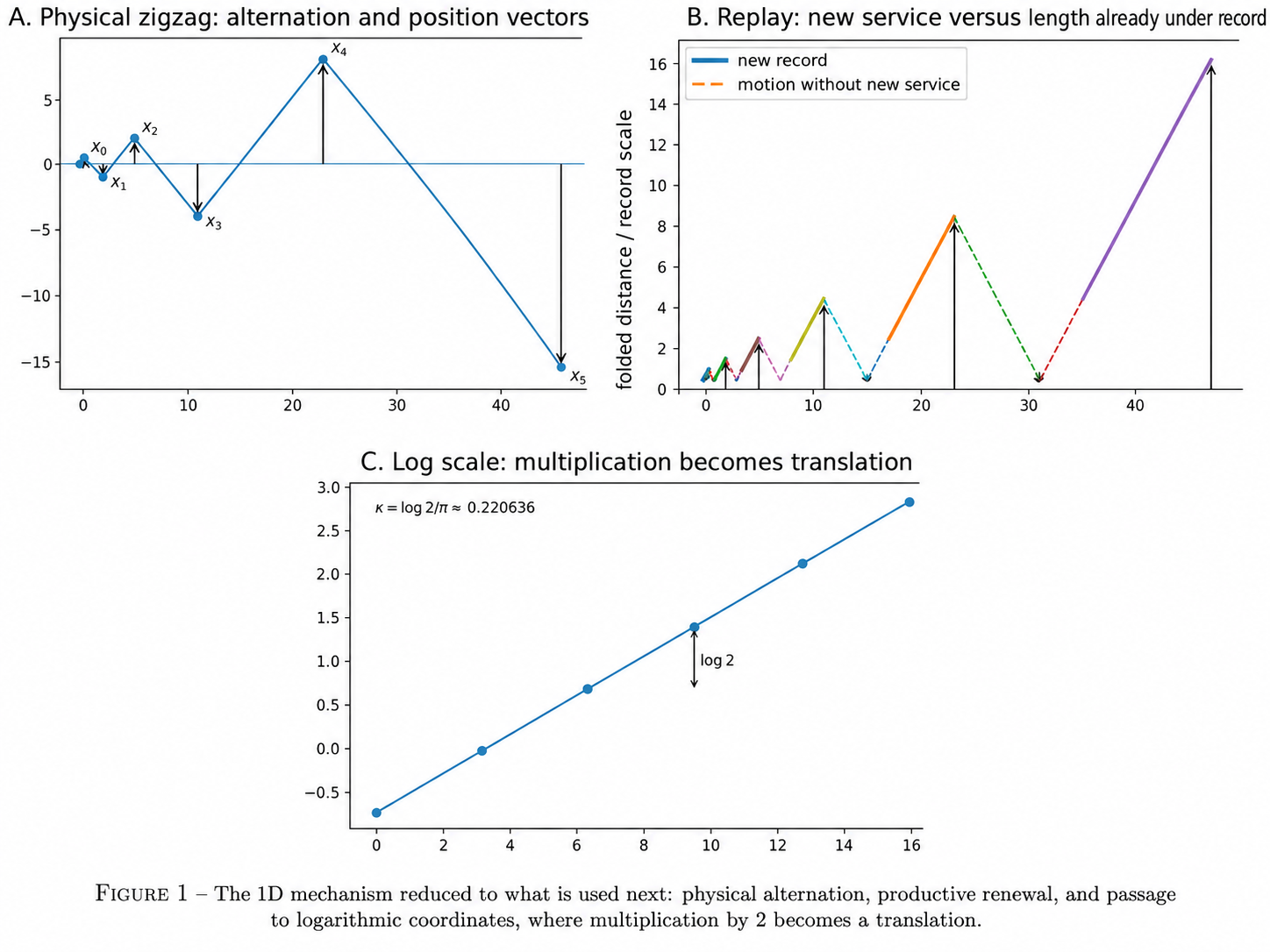}
  \caption{The one-dimensional mechanism reduced to what is used later: physical alternation, productive renewal, and the passage to logarithmic coordinates, where multiplication by two becomes a translation.}
\end{figure}

\section{Dimension two: continuous rotation of the same mechanism}

In dimension two the normal space is no longer $\Sph^0$ but the circle. For
\[
  u_\alpha=(\cos\alpha,\sin\alpha),
  \qquad
  X=R(\cos\theta,\sin\theta),
\]
one has
\[
  \langle u_\alpha,X\rangle=R\cos(\theta-\alpha).
\]
A single physical displacement therefore generates a continuum of correlated scalar cow paths. The alternation of $D=1$ becomes a continuous redistribution of service over $\Sph^1$.

\subsection{Relative equilibrium: why a spiral appears}

The relative-equilibrium theorem in~\cite[Theorem~12.8]{koch-cell} shows that a normalized state in relative equilibrium under rotation generates a logarithmic spiral. Within the stationary family we therefore write
\[
  \Gamma(\sigma)=\e^{\kappa\sigma}(\cos\sigma,\sin\sigma),
  \qquad \kappa>0.
\]
In present-scale units, the current point is $P(0)=(1,0)$ and a supplier of age $\tau$ is
\[
  P(\tau)=\e^{-\kappa\tau}(\cos\tau,-\sin\tau).
\]
The first issue is not yet to optimize $\kappa$, but to determine which old point can genuinely support the bottleneck chord together with the current point.

\subsection{First equation: the old supplier must be tangent}

If $P(\tau)$ is the active historical supplier, the line through $P(0)$ and $P(\tau)$ must be tangent to the discounted curve at $P(\tau)$. The condition
\[
  \det\bigl(P(0)-P(\tau),P'(\tau)\bigr)=0
\]
gives exactly
\[
  \e^{-\kappa\tau}=\cos\tau-\kappa\sin\tau.
\]
This equation does not choose the optimal pitch; it selects a delay compatible with a tangent chord. For the stationary Shoreline branch used below, we restrict to $\tau\in(\pi,2\pi)$ and require the chord to support the entire discounted past, with the perpendicular foot lying on the active segment. These global facts are proved in~\cite[Lemma~B.1]{koch-cell}. We write $\tau(\kappa)$ for a locally selected branch satisfying both tangency and this global support condition.

The distance from the origin to the tangent is
\[
  \mu(\kappa,\tau)=\frac{\e^{-\kappa\tau}}{\sqrt{1+\kappa^2}}.
\]
Since the normalized historical length of the stationary orbit is $\sqrt{1+\kappa^2}/\kappa$, the ratio in this family is
\[
  C_{2,\mathrm{stat}}(\kappa)
  =\frac{1+\kappa^2}{\kappa}\e^{\kappa\tau(\kappa)},
\]
where $\tau(\kappa)$ is the supporting branch just specified.

\subsection{Second equation: only now choose the pitch}

The pitch $\kappa$ is selected by minimizing this worst stationary ratio. The condition
\[
  \frac{\dd}{\dd\kappa}C_{2,\mathrm{stat}}(\kappa)=0,
\]
combined with implicit differentiation of the tangency equation, eliminates $\tau'(\kappa)$ and gives
\[
  \tan\tau=\frac{\kappa^2\tau}{1-\kappa\tau}.
\]
The logical order is therefore: historical contact first, pitch optimization second. Solving the two equations numerically yields
\[
  \kappa^*\approx0.2124695594,
  \qquad
  \tau^*\approx4.8588823628,
  \qquad
  C_{2,\mathrm{stat}}\approx13.81113517946.
\]
This is an explicit upper benchmark. The equality $\cstar_2=C_{2,\mathrm{stat}}$ remains open.

\subsection{A one-dimensional/two-dimensional dictionary}

The following table is a dictionary, not a proof that the optimal shapes coincide.

\begin{center}
\begin{tabularx}{\linewidth}{@{}p{0.25\linewidth}p{0.29\linewidth}X@{}}
\toprule
\textbf{1D} & \textbf{2D} & \textbf{Common meaning}\\
\midrule
$\Sph^0=\{-1,+1\}$ & $\Sph^1$ & direction-constraint space\\
signed position & projection $u_\alpha\cdot X$ & cow path seen by one normal\\
alternation & rotation & motion of the service allocation\\
amplitude $a_n$ & radius $R$ & physical scale\\
left/right record & support function & memory updated by maximum\\
minimum of two records & inradius $m$ & common guaranteed service\\
old record & old supplier & active memory\\
active point & active chord & bottleneck certificate\\
doubling & rotation--dilation & stationary relative growth\\
oscillator on $\Sph^0$ & oscillator on $\Sph^1$ & same mechanism, richer direction space\\
\bottomrule
\end{tabularx}
\end{center}

The essential point is already visible: each normal $\alpha$ receives a scalar projection of the same trajectory, and its record is the historical maximum of that projection. The different humps are never independent; they share one physical motion.

\begin{figure}[H]
  \centering
  \includegraphics[width=0.83\linewidth]{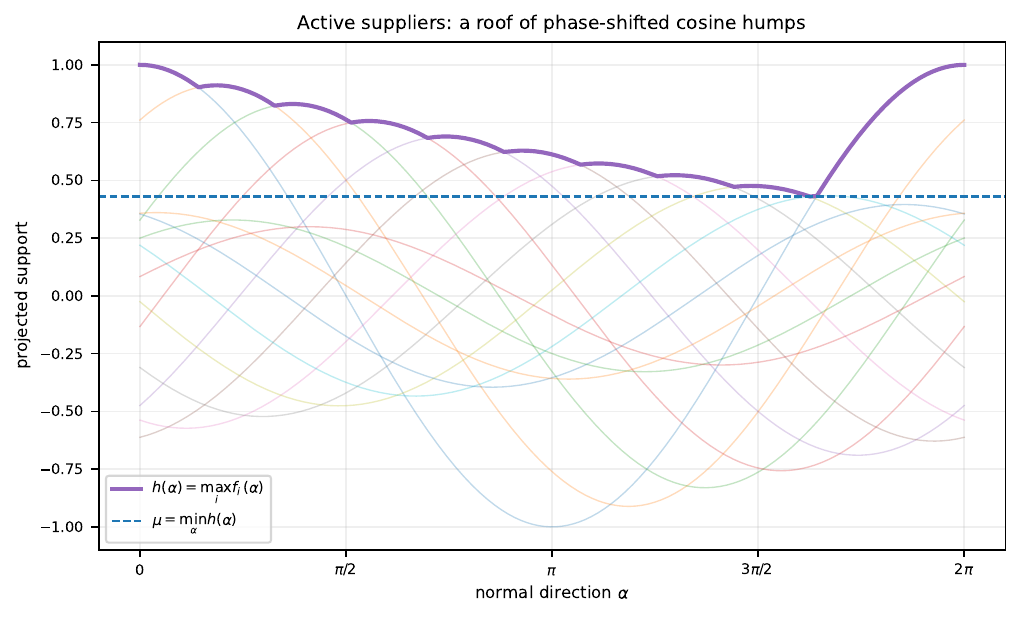}
  \caption{Dimension two: the active-supplier panel. One trajectory creates a roof of phase-shifted cosine humps; their maximum is the historical support, and the minimum of that roof is the common guaranteed service. Replacing the circle by $\Sph^{D-1}$ does not change the scalar channel, but multiplies the directions that must be served simultaneously.}
\end{figure}

\section{Exponential oscillation: the general object and its immediate consequences}

Dimensions one and two suggest separating scale growth from directional motion. We now make that definition once and derive the memory objects and gauges used throughout the rest of the article.

\begin{definition}[Exponential vector oscillator]
An exponential orbit is a trajectory
\[
  \Gamma(\sigma)=\e^{\kappa\sigma}\omega(\sigma),
  \qquad
  \omega(\sigma)\in\Sph^{D-1},
  \qquad
  \kappa>0.
\]
For every normal $u\in\Sph^{D-1}$,
\[
  \langle u,\Gamma(\sigma)\rangle
  =\e^{\kappa\sigma}\langle u,\omega(\sigma)\rangle.
\]
\end{definition}

A single moving vector thus generates a continuum of coupled scalar oscillators. Exponential growth controls relative forgetting, while motion on the sphere redistributes service among directions. No particular precession is assumed.

\begin{figure}[H]
  \centering
  \includegraphics[width=0.82\linewidth,trim=0 140bp 0 0,clip]{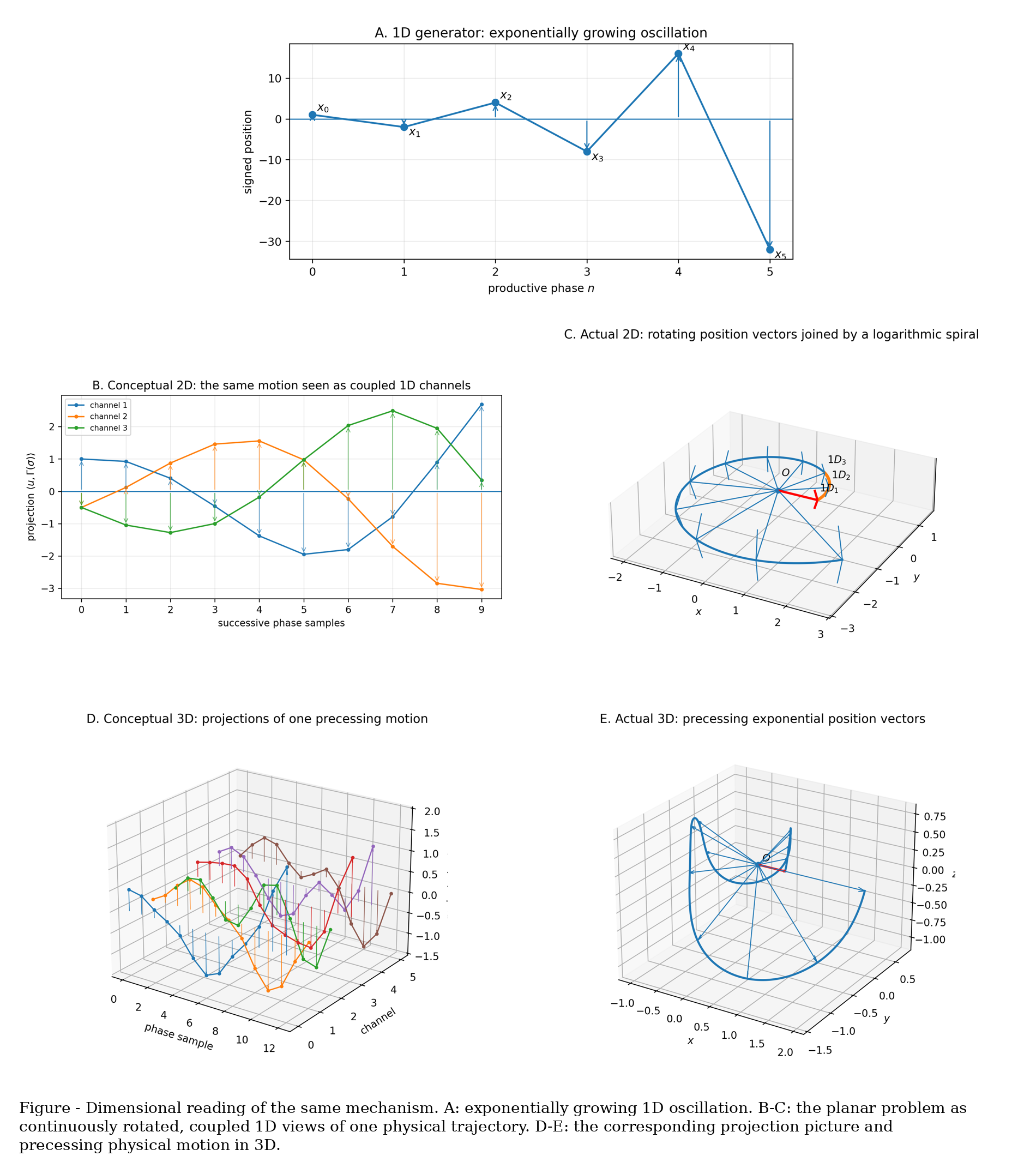}
  \caption{A recomposed dimensional reading. The one-dimensional generator is first duplicated conceptually across several directions and then recoupled by a single physical trajectory in dimensions two and three. The panels separate conceptual scalar channels from the unique physical motion that couples them.}
\end{figure}

\subsection{One physical motion, many scalar channels}

The object transported to dimension $D$ is not a collection of independent zigzags. It is the single sequence $\omega(\sigma)$ observed from many normals. The conceptual channel panels of Figure~3 are therefore only projections of one physical motion, not independent controls. This coupling replaces the separate alternation of two sides.

\subsection{Log-directional geometry}

Write $x=\e^\rho\omega$ with $\omega\in\Sph^{D-1}$. The Euclidean metric becomes
\[
  \dd s^2=\e^{2\rho}\bigl(\dd\rho^2+\dd s_{\Sph^{D-1}}^2\bigr).
\]
The quotient by scale therefore has the natural geometry $\R_\rho\times\Sph^{D-1}$. For the hyperplane $\langle u,x\rangle=d$, on the hemisphere where $\langle u,\omega\rangle>0$,
\[
  \rho=\log d-\log\langle u,\omega\rangle
  =\log d-\log\cos\beta,
\]
where $\beta$ is the angle between $u$ and $\omega$. The kernel $-\log\cos\beta$ is intrinsic to the directional problem and is not specific to the plane.

Figure~\ref{fig:scale-quotient} shows the two useful views of the quotient: the physical stacking of scales and, after normalization, a cell of fixed logarithmic height.

\begin{figure}[H]
  \centering
  \includegraphics[width=0.98\linewidth]{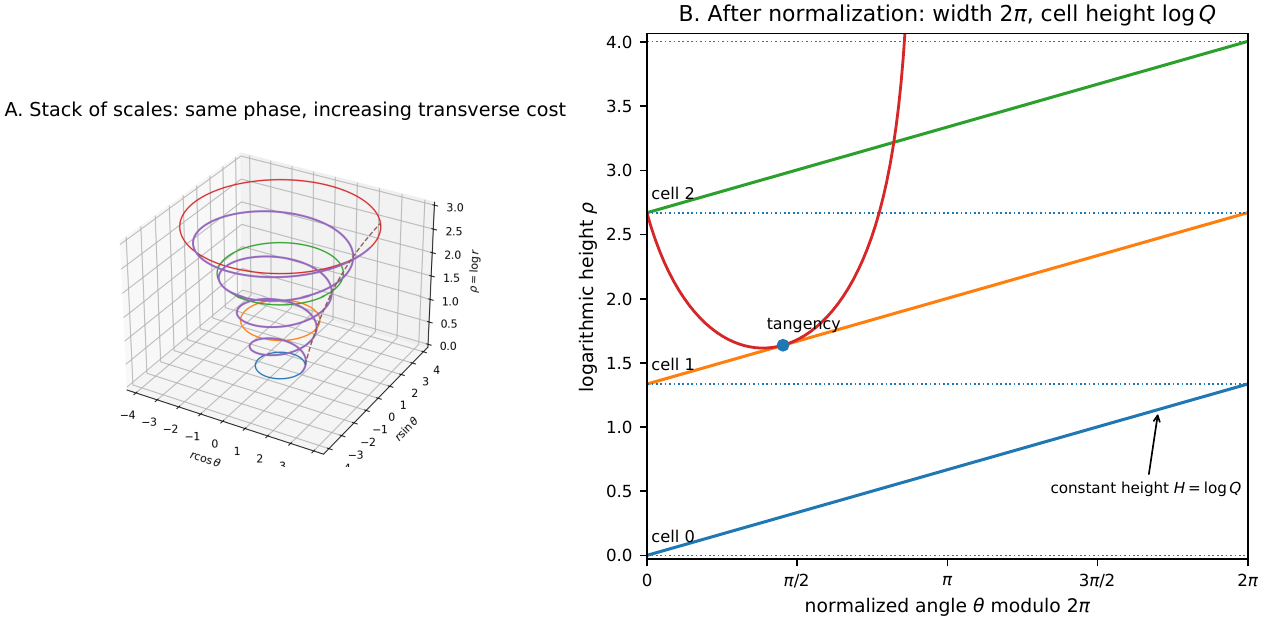}
  \caption{Scale quotient: physical stacking of copies, followed by the normalized representation in which multiplication becomes logarithmic translation.}
  \label{fig:scale-quotient}
\end{figure}

\subsection{Exponentially discounted memory}

For the orbit $\Gamma(\sigma)=\e^{\kappa\sigma}\omega(\sigma)$, a point created $\tau$ phase units earlier appears at present scale as
\[
  P_\sigma(\tau)=\e^{-\kappa\tau}\omega(\sigma-\tau).
\]
The normalized convex hull of the past is therefore
\[
  K_\sigma=\conv\{\e^{-\kappa\tau}\omega(\sigma-\tau):\tau\ge0\}.
\]
Its support function and bottleneck are
\[
  h_\sigma(n)
  =\sup_{\tau\ge0}\e^{-\kappa\tau}\langle n,\omega(\sigma-\tau)\rangle,
  \qquad
  \mu(\sigma)=\min_{\|n\|=1}h_\sigma(n).
\]
This is where the sliding-memory mechanism of~\cite[Theorem~12.2]{koch-cell} becomes an age memory: the old suppliers are precisely the points $P_\sigma(\tau)$, discounted by $\e^{-\kappa\tau}$.

If $\|\omega'(\sigma)\|=1$, the normalized historical length is
\[
  L=\int_0^\infty \e^{-\kappa\tau}\sqrt{\kappa^2+1}\,\dd\tau
  =\frac{\sqrt{1+\kappa^2}}{\kappa}.
\]
Thus the exact functional within the exponential class is
\[
  C_D[\kappa,\omega]
  =\sup_\sigma\frac{\sqrt{1+\kappa^2}}{\kappa\mu(\sigma)}.
\]
No harmonic, antipodal, or constant-$\mu$ hypothesis enters this definition.

\subsection{Dictionary of gauges}

We use consistently
\[
  t=\text{physical time},
  \qquad
  \sigma=\text{exponential phase},
  \qquad
  \tau=\text{supplier age}.
\]
When $\|\omega'\|=1$,
\[
  \Tphys(\sigma)=\frac{\sqrt{1+\kappa^2}}{\kappa}\e^{\kappa\sigma},
  \qquad
  \bar m=\frac{\kappa}{\sqrt{1+\kappa^2}}\mu,
  \qquad
  v=\frac{\kappa\omega+\omega'}{\sqrt{1+\kappa^2}}.
\]
The quantity $\bar m$ is the inradius in the Bellman gauge associated with physical time, and $v$ is the corresponding unit velocity. Whenever a dot is used below, it means differentiation with respect to logarithmic physical time $\ell=\log t_{\mathrm{phys}}$, not differentiation with respect to $t_{\mathrm{phys}}$ itself. Equivalently, $\dot f=t_{\mathrm{phys}}\,\dd f/\dd t_{\mathrm{phys}}$.

\subsection{Recursive hyperspherical parametrization}

Directional kinematics passes from one sphere to the next through the standard identity
\[
  \omega_{D+1}(t)
  =\bigl(\cos\phi_D(t),\,\sin\phi_D(t)\,\omega_D(\tau_D(t))\bigr).
\]
The derivative decomposes orthogonally:
\[
  \|\omega'_{D+1}\|^2
  =\phi_D'^2+\sin^2\phi_D\,\tau_D'^2\|\omega_D'\|^2.
\]
In coherent hyperspherical coordinates,
\[
  \|\omega_D'\|^2
  =\theta_{D-1}'^2+\sin^2\theta_{D-1}\,\theta_{D-2}'^2
  +\cdots+
  \left(\prod_{j=2}^{D-1}\sin^2\theta_j\right)\theta_1'^2.
\]
This formula is not a conjecture about the optimal shape. It is only a coherent way to construct and compare candidates on $\Sph^{D-1}$.

\subsection{The physical trajectory}

The exponential phase is denoted exclusively by $\sigma$; the symbol $t$ is reserved for physical time. In this gauge,
\[
  \Gamma_D(\sigma)=\e^{\kappa_D\sigma}\omega_D(\sigma),
  \qquad
  \|\Gamma_D'(\sigma)\|
  =\e^{\kappa_D\sigma}
   \sqrt{\kappa_D^2+\|\omega_D'(\sigma)\|^2}.
\]
When $\|\omega_D'\|=1$,
\[
  \tphys(\sigma)=\frac{\sqrt{1+\kappa_D^2}}{\kappa_D}\e^{\kappa_D\sigma}.
\]
All objects needed later are now in place: direction $\omega$, growth rate $\kappa$, suppliers $P_\sigma(\tau)$, hull $K_\sigma$, support $h_\sigma$, bottleneck $\mu$, and conversion to physical time.

\section{A dimensional staircase: what changes as $D$ grows}

The scalar mechanism does not change; what changes is the geometry required to redistribute service over a normal space of increasing dimension.

\begin{center}
\small
\begin{tabularx}{\linewidth}{@{}c p{0.17\linewidth} p{0.20\linewidth} p{0.15\linewidth} X@{}}
\toprule
$D$ & always guaranteed certificate rank & generic full-rank face & intrinsic spherical curvatures & candidate directional construction\\
\midrule
1 & $r\le1$ & $r=1$ & 0 & alternation of the two poles\\
2 & $r\le2$ & chord, $r=2$ & 0 & rotation on $\Sph^1$\\
3 & $r\le3$ & triangle, $r=3$ & 1 & precession on $\Sph^2$\\
4 & $r\le4$ & tetrahedron, $r=4$ & 2 & coupled rotations\\
5 & $r\le5$ & simplex, $r=5$ & 3 & precession / harmonic tower\\
$D$ & $r\le D$ & $r=D$ only on a nondegenerate stratum & $D-2$ & motion on $\Sph^{D-1}$ selected by the minimax problem\\
\bottomrule
\end{tabularx}
\end{center}

The certificate column anticipates the general theorem of Section~\ref{sec:active-geometry}; it does not mean that the Bellman state becomes finite-dimensional. The state remains a support function on the whole sphere. Only an instantaneous bottleneck certificate can be carried by finitely many suppliers.

Dimension two is exceptional because a rigid rotation on $\Sph^1$ can sweep every direction. In higher dimension the question becomes: how can the rotation planes themselves vary without losing exponential growth or accumulating directional deficit? This motivates precession, but no precessing family has yet been selected by the minimax problem.

A unit-speed curve on $\Sph^{D-1}$ has $D-2$ intrinsic spherical curvatures. Thus $D-1$ naturally decomposes into one radial growth rate and $D-2$ directional-curvature degrees of freedom. It must not be identified with the maximum number of historical delays, which comes from the active face rather than from kinematic coordinates.

\section{Dimension three: the first genuine precession}

For $D=3$, the directional motion lives on $\Sph^2$. A natural prototype makes latitude oscillate while azimuth rotates. Geometrically this produces a baseball-seam-type pattern on the sphere, followed by an exponentially dilated physical trajectory.

This prototype is only geometric intuition, not an optimality theorem. Section~\ref{sec:parity} proves that in odd dimension a fixed skew-symmetric generator has an invariant axis and cannot produce positive full-dimensional inradius. Precession is a natural way to vary rotation planes, not a family already proved optimal.

\subsection{From intuition to exact questions}

After $D=1$, $D=2$, and the general definition $\Gamma=\e^{\kappa\sigma}\omega$, the issue is no longer to guess another trajectory formula. One must understand what constraints any exponential orbit satisfies when its common service is weakest.

Three objects suffice:
\[
  P_\sigma(\tau)=\e^{-\kappa\tau}\omega(\sigma-\tau),
  \qquad
  K_\sigma=\conv\{P_\sigma(\tau):\tau\ge0\},
  \qquad
  \mu(\sigma)=\min_{\|n\|=1}h_{K_\sigma}(n).
\]
A bottleneck normal $n$ selects an exposed face of $K_\sigma$. The suppliers that actually carry this face are historical points of the trajectory, hence genuine ages $\tau_i$ rather than abstract variables. This observation gives the first finite compression of the active geometry.

The rest of the paper separates two levels. First come universal results valid without a shape ansatz: historical certificates, current contact at critical phases, chamber laws, tangency, pitch, coarea, and the parity obstruction. Explicit antipodal, harmonic, or precessing families are then tested against these identities.

\subsection{What is not assumed}

Nothing above proves a no-breathing regime, global stationarity, antipodality, or reduction to a primitive cycle. A cell supplied by the self-similar theorem may still contain switches, several directional cycles, and transient suppliers. The local equations below must therefore be read stratum by stratum.

The logical order is
\[
\begin{aligned}
  D=1&\longrightarrow D=2
  \longrightarrow \Gamma=\e^{\kappa\sigma}\omega\\
  &\longrightarrow (K_\sigma,h_\sigma,\mu)
  \longrightarrow \text{active geometry}.
\end{aligned}
\]
The exponential-oscillator definition is not a conjecture about the optimizer, and finite active certificates do not reduce the Bellman state to a few delays. The support remains a function on the whole sphere; only the instantaneous bottleneck has a finite witness.

\section{Active geometry and bottleneck dynamics}
\label{sec:active-geometry}

\subsection{Certificates by genuine historical suppliers}

For an exponential orbit, set
\[
  A_\sigma:=\{0\}\cup\{\e^{-\kappa\tau}\omega(\sigma-\tau):\tau\ge0\},
  \qquad
  K_\sigma=\conv A_\sigma.
\]
Adding the origin makes $A_\sigma$ compact, because suppliers of age $\tau$ converge to the origin.

The complete state is a support function on the whole sphere and is therefore infinite-dimensional. A worst direction, however, sees only one exposed face. The contact point of the centered inscribed ball lies in a face of dimension at most $D-1$, and Carathéodory's theorem reduces the certificate there to at most $D$ genuine historical suppliers.

\begin{theorem}[Historical certificate of rank at most $D$]
Suppose $\mu(\sigma)=\inrad(K_\sigma)>0$, and let $n$ minimize $h_{K_\sigma}$. Set
\[
  F_\sigma(n)=K_\sigma\cap\{x:\langle n,x\rangle=\mu(\sigma)\}.
\]
Then
\[
  \mu n\in\conv\bigl(A_\sigma\cap F_\sigma(n)\bigr).
\]
In particular, there exist genuine historical suppliers
\[
  P_0,\ldots,P_{r-1}\in A_\sigma\cap F_\sigma(n)
\]
and weights $\beta_i\ge0$, $\sum_i\beta_i=1$, such that
\[
  \mu n=\sum_{i=0}^{r-1}\beta_iP_i,
  \qquad r\le D.
\]
\end{theorem}

\begin{proof}
Since $\mu\BD\subset K_\sigma$, the point $\mu n$ belongs to $K_\sigma=\conv A_\sigma$. By Carathéodory's theorem in $\R^D$,
\[
  \mu n=\sum_j\lambda_j a_j,
  \qquad
  a_j\in A_\sigma,
  \quad \lambda_j\ge0,
  \quad \sum_j\lambda_j=1.
\]
Taking the scalar product with $n$ gives
\[
  \mu=\sum_j\lambda_j\langle n,a_j\rangle.
\]
Every term satisfies $\langle n,a_j\rangle\le h_{K_\sigma}(n)=\mu$. Hence every component of positive weight belongs to $A_\sigma\cap F_\sigma(n)$. These points lie in the affine hyperplane $\langle n,x\rangle=\mu$, of dimension at most $D-1$; Carathéodory's theorem in that affine space reduces the certificate to at most $D$ suppliers.
\end{proof}

The theorem does not say that memory is a list of $D$ ages. It says only that, at a fixed phase and for a fixed bottleneck normal, the bottleneck value has a finite witness.

\subsection{The current point at a globally critical phase}

The exact deadline inequality of~\cite[Proposition~12.5, equations~(12.11)--(12.12)]{koch-cell} gives more than mere membership of the current point in the hull. In the physical-time-normalized state
\[
  p=\frac{\gamma(t)}{t},
  \qquad
  H=\frac{K_t}{t},
\]
every trajectory of ratio at most $C$ satisfies
\[
  \langle u,p\rangle\ge1-(C-1)h_H(u).
\]
Since $p\in H$, one also has $\langle u,p\rangle\le h_H(u)$. At a worst phase, the deadline places the current point at least on the critical plane, while hull membership prevents it from crossing that plane; the point is therefore exactly on the bottleneck face.

\begin{proposition}[Current contact at the worst bottleneck]
Assume $\|\omega'\|=1$ and that $\inf_\sigma\mu(\sigma)$ is attained at $\sigma^*$. If $n$ is a bottleneck normal at that phase, then
\[
  \langle n,\omega(\sigma^*)\rangle=\mu(\sigma^*).
\]
Hence the current point belongs to the bottleneck face.
\end{proposition}

\begin{proof}
Write
\[
  A_\kappa=\frac{\sqrt{1+\kappa^2}}{\kappa}.
\]
Then
\[
  \tphys(\sigma)=A_\kappa\e^{\kappa\sigma},
  \qquad
  p=A_\kappa^{-1}\omega,
  \qquad
  H=A_\kappa^{-1}K_\sigma.
\]
At a phase realizing the worst quotient, $\min h_H=1/C$. For a minimizing normal $n$, the deadline inequality gives $\langle n,p\rangle\ge1/C$, whereas $p\in H$ gives $\langle n,p\rangle\le h_H(n)=1/C$. Equality follows, and multiplication by $A_\kappa$ gives $\langle n,\omega\rangle=\mu$.
\end{proof}

\subsection{Regular active chambers and distance to a face}

Active suppliers may switch, merge, or lose affine independence. A regular chamber is an interval on which this combinatorics remains fixed, so the equations may be differentiated without hiding switches in smooth notation.

\begin{definition}[Regular active chamber]
A phase interval is a regular active chamber of rank $r$ if there exist $C^1$ functions
\[
  \tau_0\equiv0,
  \qquad \tau_i>0\ (i\ge1),
  \qquad n,\mu,
  \qquad \beta_i>0,
\]
such that
\[
  P_0=\omega(\sigma),
  \qquad
  P_i=\e^{-\kappa\tau_i}\omega(\sigma-\tau_i),
\]
the points $P_i$ are affinely independent, and
\[
  \langle n,P_i\rangle=\mu,
  \qquad
  \mu n=\sum_{i=0}^{r-1}\beta_iP_i,
  \qquad
  \sum_i\beta_i=1,
  \qquad
  \|n\|=1,
\]
with $\mu n\in\relint\conv(P_0,\ldots,P_{r-1})$. For $i\ge1$, the age contact is interior and nondegenerate. In addition, the candidate face is required to be globally supporting throughout the chamber:
\[
  \langle n,\e^{-\kappa\tau}\omega(\sigma-\tau)\rangle\le\mu
  \qquad\text{for every }\tau\ge0.
\]
Rank loss, boundary contacts, degenerate age maxima, and switches are excluded.
\end{definition}

For affinely independent points $P_0,\ldots,P_{r-1}$, set
\[
  V=(P_1-P_0\mid\cdots\mid P_{r-1}-P_0),
  \qquad
  G=V^\top V,
  \qquad
  c=V^\top P_0.
\]
The point of their affine span closest to the origin is
\[
  y=P_0-VG^{-1}c,
  \qquad
  \|y\|^2=\|P_0\|^2-c^\top G^{-1}c.
\]
If its barycentric coordinates are strictly positive, this is the distance to the active face. Otherwise the minimum is carried by a lower-rank stratum. The determinant formula for a full simplex is therefore only the generic case $r=D$.

\subsection{The bottleneck envelope law}

Without fresh injection, every old supplier contracts at rate $\kappa$, so the bottleneck would decay like $\e^{-\kappa\sigma}$. The current point is the only supplier that does not age in this way. The next law is the exact balance between aging and present injection; it is the differential analogue at the bottleneck of the Bellman update $h^+=\max(ah,h_c)$.

\begin{proposition}[Envelope law in a regular chamber]
Within a regular active chamber,
\[
  \mu'+\kappa\mu
  =\beta_0\,n\cdot(\kappa\omega+\omega').
\]
In logarithmic physical time $\ell=\log t_{\mathrm{phys}}$, with
\[
  \bar m=\frac{\kappa}{\sqrt{1+\kappa^2}}\mu,
  \qquad
  v=\frac{\kappa\omega+\omega'}{\sqrt{1+\kappa^2}},
\]
one has
\[
  \frac{\dd\bar m}{\dd\ell}+\bar m=\beta_0n\cdot v,
  \qquad \ell=\log t_{\mathrm{phys}}.
\]
Equivalently,
\[
  t_{\mathrm{phys}}\frac{\dd\bar m}{\dd t_{\mathrm{phys}}}+\bar m
  =\beta_0n\cdot v.
\]
\end{proposition}

\begin{proof}
For an old supplier $i\ge1$, interior stationarity in age of
\[
  \tau\longmapsto \e^{-\kappa\tau}n\cdot\omega(\sigma-\tau)
\]
gives
\[
  n\cdot(\kappa\omega_i+\omega_i')=0.
\]
Differentiating $P_i=\e^{-\kappa\tau_i}\omega(\sigma-\tau_i)$, this tangency implies, independently of $\tau_i'$,
\[
  n\cdot P_i'=-\kappa\mu,
  \qquad i\ge1.
\]
For the current point, $P_0'=\omega'$. Differentiating $n\cdot P_i=\mu$ and averaging with the weights $\beta_i$ yields
\[
  \mu'
  =n'\cdot\sum_i\beta_iP_i+
    \sum_i\beta_i n\cdot P_i'
  =\sum_i\beta_i n\cdot P_i',
\]
because $\sum_i\beta_iP_i=\mu n$ and $n'\cdot n=0$. Therefore
\[
  \mu'=\beta_0n\cdot\omega'-(1-\beta_0)\kappa\mu.
\]
Since the current point is active, $n\cdot\omega=\mu$, which gives the formula. The conversion to $\ell$ uses $\dd\ell/\dd\sigma=\kappa$.
\end{proof}

In dimension two, for an active chord
\[
  p=\bar mn+s_-t,
  \qquad
  y=\bar mn+s_+t,
  \qquad
  L=s_+-s_-,
\]
the current weight is $\beta_0=s_+/L$. The general law becomes
\[
  \frac{\dd\bar m}{\dd\ell}
  =-\bar m+\frac{s_+}{L}(n\cdot v),
\]
which agrees with the planar two-contact chamber in~\cite[Section~C.3]{koch-cell}.

\begin{remark}[Boundary certificates]
Definition~8.3 already assumes $\beta_i>0$, so the former corollary was circular. The useful nontrivial statement concerns limiting boundary strata: if a smooth limiting certificate has $\beta_0=0$, the same envelope calculation gives $\mu'=-\kappa\mu<0$. Hence such a boundary certificate cannot carry a differentiable local minimum of the bottleneck.
\end{remark}

\subsection{Tangency, pitch, and age}

The tangency of every old supplier is
\[
  n\cdot(\kappa\omega_i+\omega_i')=0,
  \qquad i\ge1.
\]
Changing $\kappa$ moves old suppliers more than young ones. To first order, the total effect on the bottleneck is a barycentric mean of the active ages. Stationarity of the ratio fixes this mean exactly.

\begin{proposition}[Pitch identity]
Consider a regular chamber and vary $\kappa$ while keeping the directional shape $\omega$ fixed. If
\[
  C(\kappa)=\frac{\sqrt{1+\kappa^2}}{\kappa\mu(\kappa)}
\]
is stationary in $\kappa$, then, with $\tau_0=0$,
\[
  \sum_i\beta_i\tau_i=\frac{1}{\kappa(1+\kappa^2)}.
\]
The same identity holds under joint optimization of shape parameters whenever the usual envelope-theorem hypotheses are satisfied.
\end{proposition}

\begin{proof}
Stationarity of $\log C$ gives
\[
  0=\frac{\kappa}{1+\kappa^2}-\frac1\kappa-\frac{\mu_\kappa}{\mu},
  \qquad
  \frac{\mu_\kappa}{\mu}=-\frac{1}{\kappa(1+\kappa^2)}.
\]
At fixed shape, $\partial_\kappa P_i=-\tau_iP_i$. The envelope theorem at the active foot yields
\[
  \mu_\kappa
  =\sum_i\beta_i n\cdot(-\tau_iP_i)
  =-\mu\sum_i\beta_i\tau_i.
\]
Combining the two identities proves the claim.
\end{proof}

Every active supplier satisfies
\[
  \mu=n\cdot P_i\le\|P_i\|=\e^{-\kappa\tau_i},
  \qquad
  \tau_i\le\frac1\kappa\log\frac1\mu.
\]
Sliding memory supplies a universal horizon, while this inequality gives an additional age bound directly in the exponential gauge.

\subsection{No breathing and just-in-time saturation}

The ratio requires only a lower floor for $\mu$; it does not prevent $\mu$ from breathing above its minimum. The no-breathing regime is the stronger equality in which no common reserve is stored and later consumed.

For an arbitrary exponential orbit,
\[
  C=\sup_\sigma\frac{\sqrt{1+\kappa^2}}{\kappa\mu(\sigma)}
  =\frac{\sqrt{1+\kappa^2}}{\kappa\inf_\sigma\mu(\sigma)}.
\]
Hence
\[
  \inf_\sigma\mu(\sigma)=\frac{\sqrt{1+\kappa^2}}{\kappa C},
\]
but nothing forces $\mu$ to be constant. We call
\[
  \mu(\sigma)\equiv\frac{\sqrt{1+\kappa^2}}{\kappa C}
\]
the no-breathing regime at ratio $C$. Since
\[
  m_{\mathrm{phys}}(\sigma)=\e^{\kappa\sigma}\mu(\sigma),
  \qquad
  \tphys(\sigma)=\frac{\sqrt{1+\kappa^2}}{\kappa}\e^{\kappa\sigma},
\]
one has exactly
\[
  \text{no breathing at ratio }C
  \quad\Longleftrightarrow\quad
  m_{\mathrm{phys}}(t)-\frac{t}{C}\equiv0.
\]
This is the precise bridge to just-in-time Bellman saturation. Nothing here proves that every optimal orbit must be no-breathing.

\subsection{Frenet frame and a regular delay system in dimension three}

In dimension three, a unit-speed curve $\omega$ on $\Sph^2$ has the frame
\[
  e_0=\omega,
  \qquad
  e_1=\omega',
  \qquad
  e_2=e_0\times e_1,
\]
with
\[
  e_0'=e_1,
  \qquad
  e_1'=-e_0+\kappa_g e_2,
  \qquad
  e_2'=-\kappa_g e_1.
\]
This representation converts static KKT conditions into supplier-delay equations. It remains local: rank loss, boundary contacts, and switches require gluing across strata.

\begin{proposition}[Regular implicit delay systemin $D=3$]
Consider a regular chamber of rank three and set
\[
  V_1=P_1-P_0,
  \qquad
  V_2=P_2-P_0,
  \qquad
  N=V_1\times V_2,
  \qquad
  n=\frac{N}{\|N\|},
\]
with a fixed orientation in the chamber. Then
\[
  n'
  =\frac{(I-nn^\top)\bigl(V_1'\times V_2+V_1\times V_2'\bigr)}{\|N\|}.
\]
For an old contact $i\in\{1,2\}$, set
\[
  \omega_i=\omega(\sigma-\tau_i),
  \qquad
  t_i=\omega'(\sigma-\tau_i),
  \qquad
  q_i=\omega_i\times t_i,
  \qquad
  g_i=\kappa_g(\sigma-\tau_i).
\]
If the age maximum is nondegenerate, then away from switches,
\[
  1-\tau_i'
  =-\frac{n'\cdot(\kappa\omega_i+t_i)}{-(1+\kappa^2)\mu\e^{\kappa\tau_i}+g_i(n\cdot q_i)}.
\]
The denominator equals $\e^{\kappa\tau_i}\partial_{\tau\tau}f_i(\tau_i)$ for
\[
  f_i(\tau)=\e^{-\kappa\tau}n\cdot\omega(\sigma-\tau).
\]
After substituting $P_i'$ into the formula for $n'$, these relations form a finite linear system for $(\tau_1',\tau_2')$. They determine the delay derivatives whenever the corresponding chamber Jacobian is nonsingular; the displayed equations are therefore implicit rather than already solved in closed form.
\end{proposition}

\begin{proof}
The formula for $n'$ is the derivative of the normalized vector $N/\|N\|$; the projection $I-nn^\top$ removes the forbidden radial component. The tangency of an old contact is
\[
  n\cdot(\kappa\omega_i+t_i)=0.
\]
Differentiating with respect to $\sigma$ and using
\[
  \frac{\dd t_i}{\dd(\sigma-\tau_i)}=-\omega_i+g_iq_i
\]
gives
\[
  0=n'\cdot(\kappa\omega_i+t_i)
    +(1-\tau_i')n\cdot(-\omega_i+\kappa t_i+g_iq_i).
\]
Tangency also gives
\[
  n\cdot\omega_i=\mu\e^{\kappa\tau_i},
  \qquad
  n\cdot t_i=-\kappa\mu\e^{\kappa\tau_i}.
\]
Thus the second factor is
\[
  -(1+\kappa^2)\mu\e^{\kappa\tau_i}+g_i(n\cdot q_i),
\]
which proves the formula. Direct calculation of $\partial_{\tau\tau}f_i$ gives the same denominator.
\end{proof}

\subsection{Stratified KKT conditions and global support}

Contact equations alone are insufficient: they may describe a local face that is not actually extreme. The global inequality saying that no other age crosses the same supporting plane is indispensable.

In a regular chamber of rank $r\le D$,
\[
  \tau_0=0,
  \qquad
  P_0=\omega(\sigma),
  \qquad
  P_i=\e^{-\kappa\tau_i}\omega(\sigma-\tau_i),
\]
\[
  \|n\|=1,
  \qquad
  n\cdot P_i=\mu,
  \qquad
  \mu n=\sum_{i=0}^{r-1}\beta_iP_i,
  \qquad
  \beta_i>0,
  \qquad
  \sum_i\beta_i=1,
\]
and, for old contacts,
\[
  n\cdot(\kappa\omega_i+\omega_i')=0.
\]
These equalities must be accompanied by
\[
  n\cdot\bigl(\e^{-\kappa\tau}\omega(\sigma-\tau)\bigr)\le\mu
  \qquad\text{for all }\tau\ge0.
\]
This inequality is now part of Definition~8.3 rather than an afterthought; it is repeated here because it is essential in numerical certification. The KKT conditions are necessary on the chosen stratum. They do not certify orbit closure, equalization of worst phases, or global optimality.

\subsection{Projected coarea in arbitrary dimension}

Each displacement has a finite total budget of projected variation. Integration over all normals separates this budget into common level, directional stock, and a variation defect; the dimension enters only through sphere areas.

\begin{proposition}[Projected coarea identity]
Assume $D\ge2$. Let $\omega_k=|\Sph^k|$. For a rectifiable curve of length $L$ and convex hull $K$,
\[
  \int_{\Sph^{D-1}}\Var(u\cdot\gamma)\,\dd u
  =\frac{2\omega_{D-2}}{D-1}L.
\]
If $m=\inrad(K)$ and
\[
  S=\int_{\Sph^{D-1}}(h_K(u)-m)\,\dd u,
  \qquad
  E=\frac{\omega_{D-2}}{D-1}L
    -\int_{\Sph^{D-1}}h_K(u)\,\dd u,
\]
then $S,E\ge0$ and
\[
  \frac{\omega_{D-2}}{D-1}L=\omega_{D-1}m+S+E.
\]
In particular,
\[
  \frac{L}{m}\ge\frac{(D-1)\omega_{D-1}}{\omega_{D-2}}.
\]
\end{proposition}

\begin{proof}
The integral-geometric projection formula gives the first identity. For each $u$,
\[
  \Var(u\cdot\gamma)\ge h_K(u)+h_K(-u).
\]
After integration and symmetry,
\[
  \int_{\Sph^{D-1}}h_K(u)\,\dd u
  \le\frac{\omega_{D-2}}{D-1}L.
\]
The decomposition follows by adding and subtracting $\omega_{D-1}m$.
\end{proof}

The term $S$ is exactly the integrated support above its minimum. The term $E$ is defined as a variation defect; interpreting it as multiplicity or retraversal is geometric intuition, not its definition.

\section{Parity and the need for genuine precession}
\label{sec:parity}

The obstruction below concerns only directional motions generated by one fixed skew-symmetric matrix. It is not a general obstruction to odd dimensions.

A rigid rotation in odd dimension always has an invariant axis. The trajectory therefore cannot push support strictly to both sides of that axis; one direction retains zero support.

\begin{proposition}[Odd-dimensional obstruction for a fixed generator]
Let $D=2m+1$ and
\[
  \omega(\sigma)=\e^{\sigma A}\omega_0,
  \qquad A^\top=-A.
\]
For every $\kappa>0$, the exponentially normalized hull
\[
  K=\conv\Bigl(\{0\}\cup\{\e^{-\kappa\tau}\omega(\sigma-\tau):\tau\ge0\}\Bigr)
\]
has zero centered inradius. Hence a full $D$-dimensional exponential orbit cannot be generated by a single fixed skew-symmetric matrix.
\end{proposition}

\begin{proof}
Every real skew-symmetric matrix of odd size has nontrivial kernel. Choose $0\neq z\in\ker A$. Then
\[
  z\cdot\e^{\sigma A}\omega_0=z\cdot\omega_0=:c.
\]
A supplier of age $\tau$ satisfies $z\cdot P(\tau)=c\e^{-\kappa\tau}$. If $c=0$, the entire orbit lies in $z^\perp$. If $c\neq0$, choose
\[
  n=-\sgn(c)\frac{z}{\|z\|}.
\]
Then $n\cdot P(\tau)\le0$ for every $\tau$, and $0\in K$, so $h_K(n)=0$. In both cases $\inrad(K)=0$.
\end{proof}

The conclusion is precise: in odd dimension, a full candidate must vary its rotation planes, amplitudes, or clocks. Precession is one natural way to do so, but the proposition selects no particular family.

\section{The antipodal subclass and the harmonic tower}

The antipodal subclass is an exact laboratory in which the two-sided mechanism of $D=1$ reappears fiber by fiber. Its even-dimensional part meets the symmetric moment curve and the Barvinok--Novik orbitopes.

\subsection{Antipodality, two half-cells, and an exact functional}

Assume
\[
  \omega(\sigma+T)=-\omega(\sigma),
  \qquad
  \rho=\e^{\kappa T}.
\]
Here $T$ is a directional half-period; the full period is $2T$ and its scale factor is $Q=\rho^2$. Define
\[
  C_\sigma
  =\conv\Bigl(\{0\}\cup
  \{\e^{-\kappa r}\omega(\sigma-r):0\le r\le T\}\Bigr).
\]
Writing an age as $jT+s$ immediately separates parity: even generations fall into the current half-cell, while odd generations fall into the opposite half-cell scaled by $\rho^{-1}$. The entire infinite memory is therefore contained in two blocks.

\begin{proposition}[Exact reduction to two half-cells]
The normalized hull of the entire past is
\[
  K_\sigma=\conv\bigl(C_\sigma\cup(-\rho^{-1}C_\sigma)\bigr).
\]
Consequently,
\[
  h_{K_\sigma}(n)
  =\max\{h_{C_\sigma}(n),\rho^{-1}h_{C_\sigma}(-n)\}.
\]
\end{proposition}

\begin{proof}
Write $r=jT+s$ with $j\ge0$ and $0\le s<T$. Antipodality gives
\[
  \e^{-\kappa r}\omega(\sigma-r)
  =\rho^{-j}(-1)^j\e^{-\kappa s}\omega(\sigma-s).
\]
If $j$ is even, this point belongs to $C_\sigma$, because $0\in C_\sigma$ and $\rho^{-j}\le1$. If $j$ is odd, it belongs to $-\rho^{-1}C_\sigma$. The two most recent half-cells are present, proving equality.
\end{proof}

Set
\[
  \ell_\sigma
  =\int_0^T \e^{-\kappa r}
   \sqrt{\kappa^2+\|\omega'(\sigma-r)\|^2}\,\dd r
\]
and
\[
  \mu_\sigma
  =\min_{\|n\|=1}
   \max\{h_{C_\sigma}(n),\rho^{-1}h_{C_\sigma}(-n)\}.
\]
The old half-cell lengths form the geometric series $\rho\ell_\sigma/(\rho-1)$, hence
\[
  C_{\mathrm{anti}}[\kappa,\omega]
  =\sup_\sigma\frac{\rho}{\rho-1}\frac{\ell_\sigma}{\mu_\sigma}.
\]

\subsection{The harmonic tower and normalized projection}

For $D=2m$, consider
\[
\begin{split}
  w_{2m}(s)=\bigl(&\cos s,\sin s,
  a_2\cos3s,a_2\sin3s,\ldots,\\
  &a_m\cos((2m-1)s),a_m\sin((2m-1)s)\bigr),
\end{split}
\]
with $a_j>0$, and set $\omega_{2m}=w_{2m}/\|w_{2m}\|$. The norm is constant. When all weights equal one, this is the symmetric moment curve of Barvinok--Novik; positive weights correspond to an invertible linear transformation.

For $D=2m+1$, introduce
\[
  w_{2m+1}(s)=\bigl(w_{2m}(s),b\sin((2m+1)s)\bigr),
  \qquad b>0,
\]
and normalize. The odd lift preserves the even tower: projection onto the first coordinates, followed by normalization, recovers the preceding motion exactly. The new mode is the smallest absent positive odd frequency, hence the first one preserving antipodality without duplicating an existing coordinate.

\begin{proposition}[Normalized projection and the first available mode]
Let $\pi_{2m}:\R^{2m+1}\to\R^{2m}$ be projection onto the first coordinates. For every $s$,
\[
  \frac{\pi_{2m}\omega_{2m+1}(s)}{\|\pi_{2m}\omega_{2m+1}(s)\|}
  =\omega_{2m}(s).
\]
The new coordinate of frequency $2m+1$ is the smallest positive odd frequency not in $\{1,3,\ldots,2m-1\}$. In particular, the first lift $2\to3$ introduces mode $3$.
\end{proposition}

\begin{proof}
The projection of $\omega_{2m+1}=w_{2m+1}/\|w_{2m+1}\|$ is $w_{2m}/\|w_{2m+1}\|$. Normalizing removes that common factor and gives $w_{2m}/\|w_{2m}\|=\omega_{2m}$. Odd frequencies are exactly those satisfying $f(s+\pi)=-f(s)$; after $1,3,\ldots,2m-1$, the first new mode is $2m+1$.
\end{proof}

This proposition explains the induction mechanism of the explicit tower. It does not say that the minimax problem chooses the first available mode.

\subsection{Canonical faces of the tower}

Choosing equally spaced phases turns the trigonometric coordinates into a discrete Fourier matrix. A highest-frequency coordinate supplies the supporting hyperplane, while completeness of residues modulo the number of points gives affine independence.

\begin{proposition}[Even canonical face with $D-1$ vertices]
Let $D=2m$ and $N=2m-1$. For
\[
  s_j=\frac{2\pi j}{N},
  \qquad j=0,\ldots,N-1,
\]
the $N=D-1$ points of the weighted even tower $\omega_{2m}(s_j)$ are affinely independent, and their convex hull is an exposed face of dimension $D-2$.
\end{proposition}

\begin{proof}
It suffices to treat the unweighted curve $U_m$, because an invertible linear transformation preserves faces and affine independence. The coordinate $\cos(Ns)$ is at most one on the curve and equals one exactly at the $s_j$, proving exposure. If $\sum_jc_j=0$ and $\sum_jc_jU_m(s_j)=0$, the frequencies $\pm1,\pm3,\ldots,\pm(N-2)$ supply every nonzero discrete Fourier coefficient modulo $N$, while the first equality supplies frequency zero. All coefficients vanish, hence $c_j=0$.
\end{proof}

The odd case is isolated by the same Fourier argument within the explicit tower; this remains a statement about the tower, not about the globally discounted bottleneck.

\begin{proposition}[Odd canonical facet with $D$ vertices]
Let $D=2m+1$, $N=D$, and consider the normalized extension
\[
  \omega_{2m+1}(s)
  =\frac{(w_{2m}(s),b\sin Ns)}{\sqrt{\|w_{2m}\|^2+b^2\sin^2Ns}}.
\]
For
\[
  s_j=\frac{\pi/2+2\pi j}{N},
  \qquad j=0,\ldots,N-1,
\]
the $N=D$ points $\omega_{2m+1}(s_j)$ are affinely independent, and their convex hull is an exposed facet of dimension $D-1$.
\end{proposition}

\begin{proof}
The last coordinate of $\omega_{2m+1}(s)$ is
\[
  \frac{b\sin Ns}{\sqrt{c_m^2+b^2\sin^2Ns}},
  \qquad c_m=\|w_{2m}\|,
\]
and reaches its maximum exactly when $\sin Ns=1$, namely at the $s_j$. The corresponding hyperplane exposes their convex hull.

For affine independence, suppose $\sum_jc_j=0$ and $\sum_jc_j\omega_{2m+1}(s_j)=0$. The normalization factor is the same at all $s_j$. The frequencies $\pm1,\pm3,\ldots,\pm(2m-1)$ represent the $N-1$ nonzero residue classes modulo the odd integer $N$; together with frequency zero from $\sum_jc_j=0$, they form the full discrete Fourier transform. Hence every $c_j$ vanishes.
\end{proof}

These faces are exact within the explicit tower. They prove neither that the discounted bottleneck face is canonical, nor that full rank is universal, nor that the tower is minimax optimal.

\subsection{The exponential-raked form}

In the even case, for a normal $n$, the function
\[
  p_n(\tau)=n\cdot\omega(\sigma-\tau)
\]
is a trigonometric polynomial with odd frequencies. The discounted support constraint
\[
  \e^{-\kappa\tau}p_n(\tau)\le\mu
\]
is equivalent to
\[
  G(\tau):=\mu\e^{\kappa\tau}-p_n(\tau)\ge0.
\]
At old interior contacts,
\[
  G(\tau_i)=G'(\tau_i)=0.
\]
This exponential-raked form is well suited to interval verification: it separates the classical trigonometric polynomial from the exponential tilt specific to Shoreline search.

\section{Candidate families and characteristic equations}

Precessing and harmonic families are now treated as explicit candidates constrained by the universal identities above.

\subsection{The generic full-rank case}

A numerical solve cannot adjust delays and pitch freely. Contact, tangency, barycenter, pitch, and global face domination must hold simultaneously.

For $r=D$, seek
\[
  \kappa,n,\mu,\tau_1,\ldots,\tau_{D-1},
  \beta_0,\ldots,\beta_{D-1},
\]
with $\tau_0=0$, $P_0=\omega(\sigma)$, and
\[
  P_i=\e^{-\kappa\tau_i}\omega(\sigma-\tau_i).
\]
The local conditions are
\[
  \|n\|=1,
  \qquad
  \langle n,P_i\rangle=\mu,
  \qquad
  \langle n,\kappa\omega_i+\omega_i'\rangle=0
  \quad(i\ge1),
\]
\[
  \mu n=\sum_i\beta_iP_i,
  \qquad
  \beta_i>0,
  \qquad
  \sum_i\beta_i=1,
  \qquad
  \sum_i\beta_i\tau_i=\frac{1}{\kappa(1+\kappa^2)}.
\]
They must be supplemented by
\[
  \langle n,\e^{-\kappa\tau}\omega(\sigma-\tau)\rangle\le\mu
  \qquad\text{for every }\tau\ge0.
\]
When $\|\omega'\|=1$, the associated ratio is
\[
  \ccand_D=\frac{\sqrt{1+\kappa^2}}{\kappa\mu}.
\]
These are necessary conditions within the chamber. One must still verify closure, equalization of worst phases, switches, and absence of a closer competing face.

\subsection{Explicit prototypes in dimensions three and four}

The $D=3$ prototype is a global curve evaluated only numerically. In an auxiliary periodic parameter $s$, set
\[
  \lambda(s)=A\sin(2s),
  \qquad
  \theta(s)=s+b_4\sin(4s),
\]
\[
  u(s)=\bigl(\cos\lambda(s)\cos\theta(s),
  \cos\lambda(s)\sin\theta(s),
  \sin\lambda(s)\bigr),
\]
with
\[
  A\approx0.6841393588,
  \qquad
  b_4\approx-0.1701762692,
  \qquad
  \kappa_s\approx0.2147414236.
\]
Screening gives a ratio near $25.600134$. This number is neither a global lower bound nor an interval-certified upper bound.

In $D=4$, the data
\[
  C_4\approx41.68613,
  \qquad
  \kappa\approx0.1069799455
\]
come only from a local active-tetrahedron configuration, with approximate delays
\[
  0,\quad3.85249085,\quad7.60682906,\quad11.70470150
\]
and barycentric weights
\[
  (0.02823217,\ 0.11600646,\ 0.29809974,\ 0.55766163).
\]
It is not presented as a globally audited orbit.

\subsection{A working atlas for dimensions one through ten}

The table intentionally mixes several proof levels, stated in the final column. Screening values are not values of $\cstar_D$.

\begin{center}
\scriptsize
\begin{tabularx}{\linewidth}{@{}c r r r r X@{}}
\toprule
$D$ & working $C_D$ & $C_D/D^{3/2}$ & $\kappa$ & $Q$ & status\\
\midrule
1 & 9.000000 & 9.000 & -- & 2.0000 & exact global optimum\\
2 & 13.811135 & 4.883 & 0.212470 & 3.79994 & stationary-family benchmark; global status open\\
3 & 25.600134 & 4.927 & 0.214741 & 3.85457 & explicit global precessing curve; uncertified numerical ratio\\
4 & 41.686132 & 5.211 & 0.106980 & -- & local active-tetrahedron data; incomplete global orbit\\
5 & 57.361155 & 5.131 & 0.090898 & 3.90288 & explicit harmonic candidate; measured, not certified\\
6 & 70.648120 & 4.807 & 0.087551 & 3.80815 & even tower; strengthened but uncertified screening\\
7 & 91.398639 & 4.935 & 0.074241 & 4.13626 & odd tower; strengthened but uncertified screening\\
8 & 106.457632 & 4.705 & 0.068124 & 4.24849 & even tower; strengthened but uncertified screening\\
9 & 132.656229 & 4.913 & 0.053607 & 3.88395 & odd tower; first strengthened uncertified estimate\\
10 & 149.096348 & 4.715 & 0.052686 & 3.72363 & even tower; first strengthened uncertified estimate\\
\bottomrule
\end{tabularx}
\end{center}

\section{N-COMP: effective computability in fixed finite dimension}

The effectively bounded cell theorem~\cite[Theorem~5.1]{koch-cell} is already dimension-independent. N-COMP is not a second, independent self-similar reduction. It combines that bounded cell theorem with dimension-free polygonalization and a semialgebraic presentation valid at every fixed dimension.

The theorem successively removes three infinities: infinite history, rectifiable-cell complexity, and continuous parameter optimization. It is not a practical solver; it says only that no logical infinity remains at a prescribed precision.

\begin{theorem}[N-COMP]
For every fixed integer $D\ge1$ and every rational $\varepsilon>0$, there is a terminating algorithm producing rationals
\[
  L_{D,\varepsilon}\le U_{D,\varepsilon}
\]
such that
\[
  L_{D,\varepsilon}\le\cstar_D\le U_{D,\varepsilon},
  \qquad
  U_{D,\varepsilon}-L_{D,\varepsilon}\le\varepsilon.
\]
In particular, $\cstar_D$ is a computable real for every fixed finite $D$. A constructive quantifier-elimination procedure can also produce a self-similar polygonal cell with algebraic coordinates and ratio at most $\cstar_D+\varepsilon$.
\end{theorem}

\begin{proof}
The case $D=1$ is exact. Assume $D\ge2$, put
\[
  \varepsilon_0=\min\{\varepsilon,1\},
  \qquad
  \eta=\varepsilon_0/16.
\]

\paragraph{1. A computable upper bound.}
Consider the cell of factor $Q=2$
\[
  e_1\to e_2\to\cdots\to e_D\to-e_1\to\cdots\to-e_D\to2e_1.
\]
Its length is
\[
  L_D^\times=(2D-1)\sqrt2+\sqrt5.
\]
Its terminal hull contains the cross-polytope $\conv\{\pm e_i\}$, whose centered inradius is $1/\sqrt D$. The previous copy guarantees level at least $1/(2\sqrt D)$ throughout the current cell, while total cost is at most $2L_D^\times$. Thus
\[
  \cstar_D\le4\sqrt D\,L_D^\times.
\]
Choose a computable rational $U_D$ strictly above this bound.

\paragraph{2. A uniformly bounded near-optimal cell.}
By definition of the infimum, there is a trajectory with ratio at most $\cstar_D+\eta<U_D+1$. Apply~\cite[Theorem~5.1]{koch-cell}, using a prescribed return gap large enough to ensure $Q\ge4$. It produces a cell $C$ with
\[
  4\le Q\le Q_{\max}(D,U_D,\eta),
  \qquad
  \mathcal R_Q(C)\le\cstar_D+2\eta
  =\cstar_D+\varepsilon_0/8,
\]
and with computably bounded normalized length. Let $m_{\min}$ be the minimum guaranteed level through the cell and rescale so that $m_{\min}=1$. Since the historical cost at every phase is at least $L(C)/(Q-1)$,
\[
  \frac{L(C)}{Q-1}\le (U_D+1)m_{\min},
\]
so after this rescaling
\[
  L(C)\le\widetilde L_{\max}:=(U_D+1)(Q_{\max}-1),
\]
a computable bound.

\paragraph{3. Dimension-free conservative polygonalization.}
We spell out the point that was previously only asserted. Let a rectifiable cell in $\R^D$ satisfy $m_C(s)\ge1$ at every phase and $L(C)\le\widetilde L_{\max}$. Partition it into arcs of length at most $h<1$ and replace each arc by its chord, keeping the cut points. For a point $z$ on the $i$th chord, compare the polygonal prefix at $z$ with the original prefix at the end $s_i$ of the corresponding arc. Every original point visited by time $s_i$ lies within distance $h$ of a cut point already available to the polygonal prefix; the same holds for the terminal cell hull. Hence the original historical hull is contained in the Minkowski sum of the polygonal historical hull and $h\BD$. Support functions therefore satisfy
\[
  h_{\mathrm{poly}}(u)\ge h_{\mathrm{orig}}(u)-h
  \qquad\text{for every }u\in\Sph^{D-1},
\]
and consequently $m_{\mathrm{poly}}(z)\ge m_C(s_i)-h$. Chords do not increase travelled length, so the polygonal cost at $z$ is at most the original cost at $s_i$. Thus, in every finite dimension,
\[
  \mathcal R_Q(C_N)\le\frac{\mathcal R_Q(C)}{1-h}.
\]
This is the dimension-free form of the conservative polygonalization argument in~\cite[Theorem~8.3]{koch-cell}.

Set $C_0=U_D+1$ and
\[
  h=\min\left\{\frac12,\frac{\varepsilon_0}{16C_0}\right\}.
\]
The resulting polygonal cell has at most
\[
  N=\left\lceil\frac{\widetilde L_{\max}}{h}\right\rceil
\]
edges and satisfies, using $1/(1-h)\le1+2h$,
\[
  \mathcal R_Q(C_N)
  \le\frac{\mathcal R_Q(C)}{1-h}
  \le\cstar_D+\varepsilon_0/4.
\]

\paragraph{4. A computable compact box.}
Normalize total cell length to one and put $\lambda=Q^{-1}$. Then
\[
  Q_{\max}^{-1}\le\lambda\le\frac14,
  \qquad
  p_0=\lambda p_n.
\]
Since $\|p_n-p_0\|=(Q-1)\|p_0\|\le1$ and $Q\ge4$, one has $\|p_0\|\le1/3$. Every vertex lies within path length one of $p_0$, so
\[
  \|p_i\|\le\frac43.
\]
The near-optimal cell therefore lies in a computable compact box.

\paragraph{5. Semialgebraic feasibility.}
For a polygonal cell $p_0,\ldots,p_n$ define, explicitly,
\[
  P:=\conv\{p_0,\ldots,p_n\},
  \qquad
  \ell_i=\|p_i-p_{i-1}\|,
  \qquad
  \sum_{i=1}^n\ell_i=1.
\]
During edge $e$ at phase $\alpha\in[0,1]$, put
\[
  z_{e,\alpha}=p_{e-1}+\alpha(p_e-p_{e-1}),
\]
\[
  K_{e,\alpha}
  =\conv\bigl(\lambda P\cup\{p_0,\ldots,p_{e-1},z_{e,\alpha}\}\bigr),
\]
\[
  A_{e,\alpha}
  =\frac{\lambda}{1-\lambda}
   +\sum_{i<e}\ell_i+\alpha\ell_e.
\]
The condition $\mathcal R_Q\le c$ is equivalent to
\[
  \forall e,\ \forall\alpha\in[0,1],\ \forall u,
  \quad \|u\|=1
  \Longrightarrow
  c\,h_{K_{e,\alpha}}(u)\ge A_{e,\alpha}.
\]
The support of a finite convex hull is a finite maximum of scalar products, and edge lengths are introduced by quadratic equations. Together with the compact box, feasibility is a first-order formula over the reals with rational coefficients and is therefore decidable exactly by quantifier elimination~\cite{basu-pollack-roy}.

\paragraph{6. Finite minimum and rational bisection.}
Taking the finite disjunction over $1\le n\le N$ gives an attained minimum $W_{D,N}$ with
\[
  \cstar_D\le W_{D,N}\le\cstar_D+\varepsilon_0/4.
\]
A rational bisection using the feasibility oracle encloses $W_{D,N}$ in $[a,b]$ with $b-a\le\varepsilon_0/2$. Hence
\[
  a-\varepsilon_0/4\le\cstar_D\le b,
\]
and the interval has width at most $3\varepsilon_0/4\le\varepsilon$. Constructive quantifier elimination also returns an algebraic witness below every feasible rational threshold.
\end{proof}

\begin{remark}
N-COMP is a computability theorem, not a competitive numerical solver. The covering numbers in~\cite[Theorem~5.1]{koch-cell} and the complexity of quantifier elimination are enormous.
\end{remark}

\section{Numerical certification and falsification}

Dimensional candidates must be tested by searches that do not impose the candidates' own form.

\subsection{Free polygonal cells}

For fixed $D$ and $N$ edges, optimize freely over
\[
  p_0,\ldots,p_{N-1}\in\R^D,
  \qquad Q>1,
  \qquad p_N=Qp_0,
\]
while evaluating the true online quotient on every prefix. Floating-point optimization may discover competitors; it does not certify their ratios.

\subsection{Certified directional nets}

A random cloud of directions never certifies an inradius. If $H\subset R\BD$ and $\mathcal N_\varepsilon$ is an $\varepsilon$-net of $\Sph^{D-1}$, the support function is $R$-Lipschitz and
\[
  \inrad(H)\ge\min_{v\in\mathcal N_\varepsilon}h_H(v)-R\varepsilon.
\]
The remainder $R\varepsilon$ is indispensable. Net cost typically grows like $\varepsilon^{-(D-1)}$; this is a quantitative obstacle, not a logical one.

\subsection{Exploratory screening through $D=10$}

For $D=2m$, the screening uses the even tower
\[
  w(s)=\bigl(a_1\cos s,a_1\sin s,\ldots,
  a_m\cos((2m-1)s),a_m\sin((2m-1)s)\bigr),
\]
with $a_1=1$. For $D=2m+1$, an axial component $b\sin((2m+1)s)$ is added. After normalization and, when necessary, reparametrization by spherical arclength, set
\[
  Q=\e^{\kappa T_{\mathrm{per}}}.
\]
The search alternates phase sampling, normal sampling, and local bottleneck refinement. It remains screening until phases and directions are covered with outward error bounds.

The retained weights for $D=5,\ldots,10$ are:

\begin{center}
\scriptsize
\begin{tabularx}{\linewidth}{@{}c X r r r@{}}
\toprule
$D$ & harmonic weights, $a_1=1$ & $Q$ & $\kappa$ & measured $C$\\
\midrule
5 & $(1,0.638966,0.623845)$; final component axial & 3.90288 & 0.090898 & 57.3612\\
6 & $(1,0.522860,0.461342)$ & 3.80815 & 0.087551 & 70.6481\\
7 & $(1,0.633023,0.499238,0.490921)$; final component axial & 4.13626 & 0.074241 & 91.3986\\
8 & $(1,0.552117,0.477382,0.463268)$ & 4.24849 & 0.068124 & 106.4576\\
9 & $(1,0.719643,0.518417,0.511247,0.514496)$; final component axial & 3.88395 & 0.053607 & 132.6562\\
10 & $(1,0.710321,0.465677,0.410313,0.403021)$ & 3.72363 & 0.052686 & 149.0963\\
\bottomrule
\end{tabularx}
\end{center}

\begin{figure}[H]
  \centering
  \includegraphics[width=0.92\linewidth]{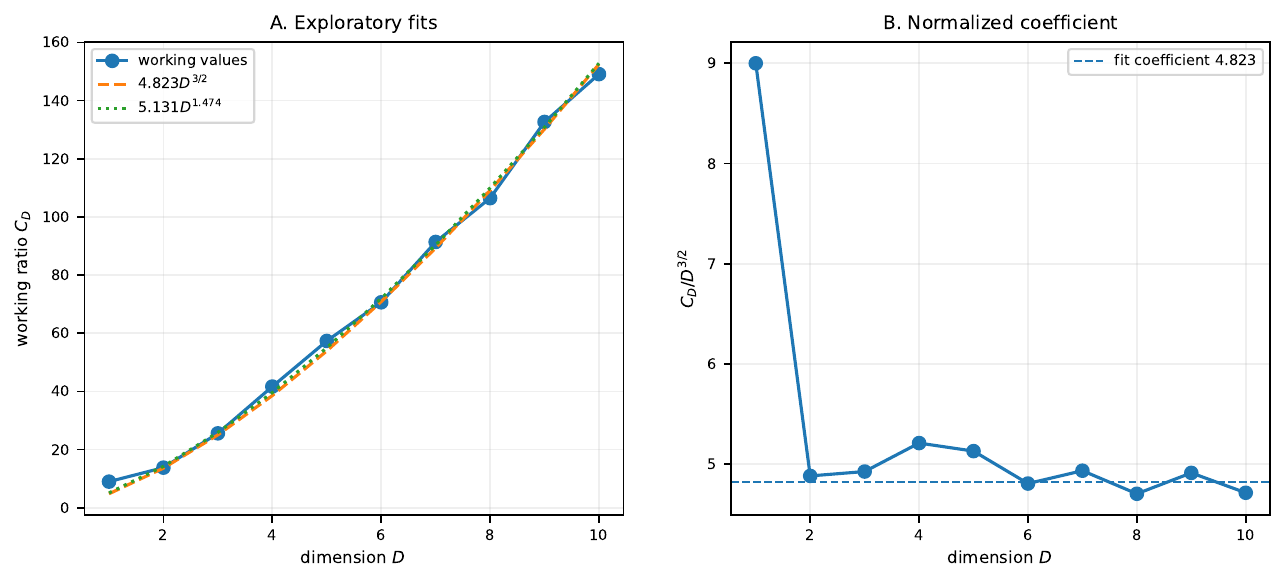}
  \caption{Working atlas for $D=1,\ldots,10$. The fits and normalized coefficient concern candidate values; they are not certified estimates of $\cstar_D$.}
\end{figure}

Over $D=2,\ldots,10$, the working values give the exploratory fits
\[
  \cwork_D\simeq4.823D^{3/2},
  \qquad
  \cwork_D\simeq5.131D^{1.474}.
\]
Over $D=5,\ldots,10$,
\[
  \frac16\sum_{D=5}^{10}D\kappa_D\simeq0.509,
  \qquad
  \frac16\sum_{D=5}^{10}Q_D\simeq3.951.
\]
These are small-dimensional diagnostics, not estimates of limiting constants.

\section{The numerical atlas and the separation of two certificates}

To establish
\[
  \ccand_D-\cstar_D\le\delta,
\]
one needs two independent certificates:
\[
  L_D^{\mathrm{global}}\le\cstar_D,
  \qquad
  \ccand_D\le U_D^{\mathrm{cand}},
  \qquad
  U_D^{\mathrm{cand}}-L_D^{\mathrm{global}}\le\delta.
\]
The first bounds the global optimum from below; the second certifies only the chosen candidate. Current screening supplies explicit curves, not yet both certificates.

For $D=1$, the two coincide and give $\cstar_1=9$. For $D=2$, the known global lower bound and the spiral benchmark remain separate. For $D=3$ and $D=5$, global candidates are available without corresponding lower bounds. For $D=6$ through $10$, the displayed values are reproducible floating-point ratios, not rigorous upper bounds.

The certification route is clear: cover phases and normals by boxes with Lipschitz remainders to obtain $U_D^{\mathrm{cand}}$; independently, use N-COMP or a Bellman branch-and-bound to obtain $L_D^{\mathrm{global}}$.

\section{Large dimension and comparison with the literature}

The literature places the minimax value at order
\[
  \cstar_D=\Theta(D^{3/2})
\]
up to multiplicative constants~\cite{antoniadis-hyperplane,bansal-cowpath}. The new question is therefore not the exponent, but the structure of candidates and the possible existence of an asymptotic constant.

Our screening studies
\[
  \frac{\ccand_D}{D^{3/2}},
  \qquad
  D\kappa_D,
  \qquad
  Q_D.
\]
Six dimensions do not support serious estimation of limits; they motivate only the following conjectures.

\begin{remark}[Limits within a harmonic/precessing family]
For the best candidates in a precisely specified class, one may conjecture positive constants $c_{\mathrm{harm}}$, $\lambda_{\mathrm{harm}}$, and $Q_\infty$ such that
\[
  \frac{C_D^{\mathrm{harm}}}{D^{3/2}}\to c_{\mathrm{harm}},
  \qquad
  D\kappa_D^{\mathrm{harm}}\to\lambda_{\mathrm{harm}},
  \qquad
  Q_D^{\mathrm{harm}}\to Q_\infty.
\]
The small dimensions in this article are not proposed as estimates of these limits.
\end{remark}

\begin{remark}[A minimax asymptotic constant]
A stronger conjecture is the existence of
\[
  \lim_{D\to\infty}\frac{\cstar_D}{D^{3/2}}=c_\infty>0.
\]
If a harmonic/precessing family proved asymptotically optimal, then $c_\infty=c_{\mathrm{harm}}$.
\end{remark}

\subsection{The static proxy in dimension three}

The static problem of the shortest closed curve whose convex hull contains the unit ball of $\R^3$ has a baseball-seam-type solution of length $4\pi$~\cite{ghomi-wenk}. This does not give the online Shoreline value in $D=3$, but it is a consistency check for the appearance of precessing geometry.

\section{Open questions}

The following obstacles remain logically separate from the preceding results.

\begin{enumerate}[label=\arabic*.]
\item \textbf{Productive synchronization.} Can one construct a globally monotone productive phase that preserves true first passages, including angular returns, multiple winding, and transient suppliers?
\item \textbf{Reduction to a primitive cycle.} Can a near-optimal macrocell containing several directional cycles be shortened while retaining the full Bellman state?
\item \textbf{Global no breathing.} Can a nonstationary cycle of storage and handoff have positive minimax yield?
\item \textbf{Stationarization.} Can a near-optimal orbit be replaced without loss by a relative equilibrium in the scale quotient?
\item \textbf{Optimal precession.} Nothing yet proves that the minimax problem selects the harmonic tower or a particular nested precession.
\item \textbf{Antipodality.} It simplifies memory exactly within a subclass but is not globally forced.
\item \textbf{Parity of faces.} Canonical tower faces are explicit; the global classification of discounted bottleneck faces remains open.
\item \textbf{A practical algorithm.} How can N-COMP be turned into a certified branch-and-bound exploiting sliding memory, active rank, age, stratified KKT conditions, and directional nets?
\end{enumerate}

\section{Conclusion}

Passing from dimension one to arbitrary finite dimension does not mean replacing a simple curve by a complicated one. It means isolating the mechanisms that survive the change of dimension:
\[
  \text{exponential growth}
  +\text{ memory by maximum}
  +\text{ oscillation of service on }\Sph^{D-1}.
\]

Dimension one realizes this mechanism by alternation and gives $\cstar_1=9$; dimension two turns the same principle into continuous rotation; dimension three is the first case in which a fixed skew-symmetric generator is impossible and the rotation planes must vary.

Without harmonic assumptions, the bottleneck of an exponential orbit has a certificate supported by at most $D$ historical suppliers. In a regular chamber, this certificate yields the envelope law, supplier tangencies, the pitch identity, an age bound, and, in $D=3$, a finite regular implicit delay system. These results do not reduce the functional Bellman state to a finite state and do not prove global stationarity.

The antipodal subclass has an exact two-half-cell memory. The harmonic tower respects normalized projection between dimensions, introduces successive odd modes, and has explicit canonical faces in both even and odd cases. The feature specific to Shoreline search is the coupling of this geometry to exponential discounting and the true online ratio.

Independently of any structural selection, N-COMP proves that $\cstar_D$ is computable for every fixed finite dimension. The stronger structural problem remains: rule out profitable storage-and-handoff cycles, understand switches, and determine whether the minimax problem ultimately selects a particular precessing class.


\begin{thebibliography}{99}

\bibitem{koch-cell}
F.~Koch,
\emph{Bellman Search in Arbitrary Finite Dimension: A Self-Similar Cell Theorem and Effective Computability of Planar Shoreline Search},
Version~8.3, Zenodo, 2026.
\href{https://doi.org/10.5281/zenodo.22152752}{doi:10.5281/zenodo.22152752}.

\bibitem{bellman-dp}
R.~Bellman,
\emph{Dynamic Programming},
Princeton University Press, 1957.

\bibitem{beck-newman}
A.~Beck and D.~J.~Newman,
``Yet More on the Linear Search Problem,''
\emph{Israel Journal of Mathematics} 8 (1970), 419--429.

\bibitem{baeza-yates}
R.~A.~Baeza-Yates, J.~C.~Culberson, and G.~J.~E.~Rawlins,
``Searching in the Plane,''
\emph{Information and Computation} 106 (1993), 234--252.

\bibitem{finch-zhu}
S.~R.~Finch and L.-Y.~Zhu,
``Searching for a Shoreline,''
arXiv:math/0501123, 2005.

\bibitem{finch-spiral}
S.~R.~Finch,
``The Logarithmic Spiral Conjecture,''
arXiv:math/0501133, 2005.

\bibitem{gal-chazan}
S.~Gal and D.~Chazan,
``On the Optimality of the Exponential Functions for Some Minimax Problems,''
\emph{SIAM Journal on Applied Mathematics} 30 (1976), 324--348.

\bibitem{antoniadis-hyperplane}
A.~Antoniadis, R.~Hoeksma, S.~Kisfaludi-Bak, and K.~Schewior,
``Online Search for a Hyperplane in High-Dimensional Euclidean Space,''
\emph{Information Processing Letters} 177 (2022), 106262.

\bibitem{bansal-cowpath}
N.~Bansal, J.~Kuszmaul, and W.~Kuszmaul,
``A Nearly Tight Lower Bound for the $d$-Dimensional Cow-Path Problem,''
\emph{Information Processing Letters} 182 (2023), 106389.

\bibitem{ghomi-wenk}
M.~Ghomi and J.~Wenk,
``Shortest Closed Curve to Contain a Sphere in Its Convex Hull,''
\emph{Bulletin of the London Mathematical Society} 56 (2024), 2472--2482.

\bibitem{temerev}
A.~Temerev,
``A New Unconditional Lower Bound for Shoreline Search,''
arXiv:2608.18362, 2026.

\bibitem{barvinok-novik}
A.~Barvinok and I.~Novik,
``A Centrally Symmetric Version of the Cyclic Polytope,''
\emph{Discrete \& Computational Geometry} 39 (2008), 76--99.

\bibitem{barvinok-lee-novik-neighborliness}
A.~Barvinok, S.~J.~Lee, and I.~Novik,
``Neighborliness of the Symmetric Moment Curve,''
\emph{Mathematika} 59 (2013), 223--249.

\bibitem{barvinok-lee-novik-faces}
A.~Barvinok, S.~J.~Lee, and I.~Novik,
``Centrally Symmetric Polytopes with Many Faces,''
\emph{Israel Journal of Mathematics} 195 (2013), 457--472.

\bibitem{vinzant}
C.~Vinzant,
``Edges of the Barvinok--Novik Orbitope,''
\emph{Discrete \& Computational Geometry} 46 (2011), 479--487.

\bibitem{king-mixon}
E.~J.~King and D.~G.~Mixon,
``Short Spherical $t$-Design Curves,''
arXiv:2607.15386, 2026.

\bibitem{nudelman}
A.~A.~Nudelman,
``Isoperimetric Problems for the Convex Hulls of Polygonal Lines and Curves in Multidimensional Spaces,''
\emph{Mathematics of the USSR-Sbornik} 25 (1975), 276--294.

\bibitem{schneider}
R.~Schneider,
\emph{Convex Bodies: The Brunn--Minkowski Theory}, 2nd ed.,
Cambridge University Press, 2014.

\bibitem{santalo}
L.~A.~Santaló,
\emph{Integral Geometry and Geometric Probability}, 2nd ed.,
Cambridge University Press, 2004.

\bibitem{basu-pollack-roy}
S.~Basu, R.~Pollack, and M.-F.~Roy,
\emph{Algorithms in Real Algebraic Geometry}, 2nd ed.,
Springer, 2006.

\bibitem{weihrauch}
K.~Weihrauch,
\emph{Computable Analysis: An Introduction},
Springer, 2000.

\end{thebibliography}
\end{document}